\documentclass[a4paper,11pt]{article}
\usepackage[top=2.7cm,bottom=3cm,left=2.5cm,right=2.5cm,footskip=.7in]{geometry}

\usepackage{amsmath}
\usepackage{amssymb}
\usepackage{amsfonts,amsthm,complexity}
\usepackage[hidelinks]{hyperref}
\newtheorem{theorem}{Theorem}
\newtheorem{lemma}[theorem]{Lemma}

\newtheorem{obs}[theorem]{Observation}

\newtheorem{cor}[theorem]{Corollary}
\theoremstyle{definition}
\newtheorem{definition}{Definition}
\usepackage{thm-restate}
\usepackage{todonotes}

\usepackage{graphicx}
\usepackage{colortbl}
\usepackage{multicol}
\usepackage[normalem]{ulem}
\usepackage{authblk}

\usepackage{tikz}
\usetikzlibrary{snakes}
\usetikzlibrary{decorations.pathmorphing,matrix}

\usepackage[ruled,vlined,linesnumbered]{algorithm2e}

\newcommand{\mypath}[1]{\langle#1\rangle}

\title{Computing and Listing Avoidable Vertices and Paths}
\author[1]{Charis Papadopoulos\thanks{Research supported by the Hellenic Foundation for Research and Innovation (H.F.R.I.) under the ``First Call for H.F.R.I. Research Projects to support Faculty members and Researchers and the procurement of high-cost research equipment grant'', Project FANTA (eFficient Algorithms for NeTwork Analysis), number HFRI-FM17-431.}}
\author[2]{Athanasios E. Zisis}

\affil[1]{Department of Mathematics, University of Ioannina, Greece\\
 \texttt{charis@uoi.gr}}
\affil[2]{Department of Mathematics, University of Ioannina, Greece\\
 \texttt{athanas.zisis@gmail.com}}

\date{}

\begin{document}
\maketitle

\begin{abstract}
A simplicial vertex of a graph is a vertex whose neighborhood is a clique.
It is known that listing all simplicial vertices can be done in $O(nm)$ time or $O(n^{\omega})$ time, where $O(n^{\omega})$ is the time needed to perform a fast
matrix multiplication.
The notion of avoidable vertices generalizes the concept of simplicial vertices in the following way:
a vertex $u$ is avoidable if every induced path on three vertices with middle vertex $u$ is contained in an induced cycle.
We present algorithms for listing all avoidable vertices of a graph through the notion of minimal triangulations and common neighborhood detection.
In particular we give algorithms with running times $O(n^{2}m)$ and $O(n^{1+\omega})$, respectively.
Additionally, based on a simplified graph traversal
we propose a fast algorithm that runs in time $O(n^2 + m^2)$ and matches the corresponding running time of listing all simplicial vertices
on sparse graphs with $m=O(n)$.
Moreover, we show that our algorithms cannot be improved significantly, as we prove that under plausible complexity assumptions there is no truly subquadratic algorithm for recognizing an avoidable vertex.
To complement our results, we consider their natural generalizations of avoidable edges and avoidable paths.
We propose an $O(nm)$-time algorithm that recognizes whether a given induced path is avoidable.
\end{abstract}


\section{Introduction}
Closely related to chordal graphs is the notion of a simplicial vertex, that is a vertex whose neighborhood induces a clique.
In particular, Dirac~\cite{Dirac61} proved that every chordal graph admits a simplicial vertex.
However not all graphs contain a simplicial vertex.
Due to their importance to several algorithmic problems, such as finding a maximum clique or computing the chromatic number,
it is natural to seek for fast algorithms that list all simplicial vertices of a graph.
For doing so, the naive approach takes $O(nm)$ time, whereas the fastest algorithms take advantage of computing the square of an $n \times n$ binary matrix and run in $O(n^{\omega})$ and $O(m^{2\omega/(\omega+1)})$ time \cite{KloksKM00}. Hereafter we assume that we are given a graph $G$ on $n$ vertices and $m$ edges; currently, $\omega < 2.37286$ \cite{AlmanW21}.

A natural way to generalize the concept of simplicial vertices is the notion of an avoidable vertex.
A vertex $u$ is avoidable if either there is no induced path on three vertices with middle vertex $u$, or
every induced path on three vertices with middle vertex $u$ is contained in an induced cycle.
Thus every simplicial vertex is avoidable, however the converse is not necessarily true.
As opposed to simplicial vertices, it is known that every graph contains an avoidable vertex \cite{AboulkerCTV15,BerryB98,BerryBBS10,OHTSUKI1976622}.
Extending the notion of avoidable vertices is achieved through avoidable edges and, more general, avoidable paths.
This is accomplished by replacing the middle vertex in an induced path on three vertices by an induced path on arbitrary $k \geq 2$ vertices, denoted by $P_k$.
Beisegel et al. \cite{BeisegelCGMS19} proved first that every non-edgeless graph contains an avoidable edge, considering the case of $k=2$.
Regarding the existence of an avoidable induced path of arbitrary length, Bonamy et al. \cite{BonamyDHT20} settled a conjecture in \cite{BeisegelCGMS19}
and showed that every graph is either $P_k$-free or contains an avoidable $P_k$.
Gurvich et al. \cite{Gurvich22} strengthened the later result by showing that every induced path can be shifted in an avoidable path,
in the sense that there is a sequence of neighboring induced paths of the same length.
Although the provided proof in \cite{Gurvich22} is constructive and identifies an avoidable path given an induced path,
the proposed algorithm was not settled whether it runs in polynomial time.

Since avoidable vertices generalize simplicial vertices, it is expected that avoidable vertices find applications in further algorithmic problems.
Indeed, Beisegel et al. \cite{BeisegelCGMS19} revealed new polynomially solvable cases of the maximum weight clique problem that take advantage of the notion of avoidable vertices.
Similar to simplicial vertices, the complexity of a problem can be reduced by removing avoidable vertices, tackling the problem on the reduced graph.
It is therefore of interest to list all avoidable vertices efficiently.
If we are only interested in computing two avoidable vertices this can be done in linear time by using fast graph searches \cite{BerryBBS10,BeisegelCGMS19}. 
However, an efficient elimination process, such as deleting or removing avoidable vertices, is not enough to recursively compute the rest of the avoidable vertices.  
Thus, computing the set of all avoidable vertices requires to decide for each vertex of the graph whether it is avoidable and
a usual graph search cannot guarantee to test all vertices.

Concerning lower bounds, it is known \cite{KratschS06} that the problem of finding a triangle in an $n$-vertex graph can be reduced in $O(n^2)$
time to the problem of counting the number of simplicial vertices in an $O(n)$-vertex graph.
Moreover, Ducoffe proved that under plausible complexity assumptions computing the diameter of an AT-free graph is at least as hard as computing a simplicial vertex \cite{Ducoffe22}. %
For general graphs, the quadratic time complexity of diameter computation cannot be improved by much \cite{RodittyW13}.
We note that the currently fastest algorithms for detecting a triangle run in $O(nm)$ time and $O(n^{\omega})$ time \cite{ItaiR78}.
Notably, we show a similar lower bound for recognizing an avoidable vertex.
In particular, via a reduction form the Orthogonal-Vector problem, we prove that under the Strong Exponential-Time Hypothesis, there is no truly subquadratic algorithm for deciding whether a given vertex is avoidable.
This gives a strong evidence that our $O(nm)$- and $O(n^{\omega})$-recognition algorithms upon which are based our listing algorithms cannot be improved significantly.

A naive approach that recognizes a single vertex $u$ of a graph $G$ of whether it is avoidable or not, needs to check if all neighbors of $u$ are pairwise connected in an induced subgraph of $G$.
Thus the running time of recognizing an avoidable vertex is $O(n^3+n^2m)$ or, as explicitly stated in \cite{BeisegelCGMS19}, it can be expressed as $O(\overline{m} \cdot (n+m))$ where $\overline{m}$ is the number of edges in the complement of $G$.
Inspired by both running times, we first show that we can reduce in linear time the listing problem on a graph $G$ having $m \geq n$ and $\overline{m} \geq n$.
In a sense such a result states that graphs that are sparse ($m < n$) or dense ($\overline{m} < n$) can be decomposed efficiently to smaller connected graphs for which their complement is also connected.
Towards this direction, we give an interesting connection with the avoidable vertices on the complement of $G$.
As a result, the naive algorithms for listing all avoidable vertices take $O(n^3 \cdot m)$ and $O(n \cdot \overline{m} \cdot m)$ time, respectively.
Moreover, based on the proposed reduction we derive an optimal, linear-time, algorithm for listing all avoidable vertices on graphs having no induced path on four vertices, known as cographs.

Our main results consist of new algorithms for listing all avoidable vertices in running times comparable to the ones for listing simplicial vertices.
More precisely, we propose three main approaches that result in algorithms for listing all avoidable vertices of a graph $G$ with the following running times:
\begin{itemize}
\item $O(n^2 \cdot m)$, by using a minimal triangulation of $G$. \
A close relationship between avoidable vertices and minimal triangulation was already known \cite{BeisegelCGMS19}.
However, listing all avoidable vertices through the proposed characterization is inefficient, since one has to produce \emph{all} possible minimal triangulations of $G$.
Here we strengthen such a characterization in the sense that it provides an efficient recognition based on one particular minimal triangulation of $G$.
More precisely, we take advantage of vertex-incremental minimal triangulations that can be computed in $O(nm)$ time \cite{BerryHV06}.

\item $O(n^2 + m^2)$, by exploring structural properties on each edge of $G$. \
This approach is based on a modified, traditional breadth-first search algorithm.
Our task is to construct search trees rooted at a particular vertex that reach all vertices of a predescribed set $S$, so that every non-leaf vertex does not belong to $S$.
If such a tree exists then every path from the root to a leaf that belongs to $S$ is called an $S$-excluded path.
It turns out that $S$-excluded paths can be tested in linear time and we need to make $2m$ calls of a modified breadth-first search algorithm.

\item $O(n^{1+\omega})$, where $O(n^{\omega})$ is the running time for matrix multiplication. \
For applying a matrix multiplication approach, we contract the connected components of $G$ that are outside the closed neighborhood of a vertex.
Then we observe that a vertex $u$ is avoidable if the neighbors of $u$ are pairwise in distance at most two in the contracted graph.
As the distance testing can be encapsulated by the square of its adjacency matrix, we deduce an algorithm that takes advantage of a fast matrix multiplication.
\end{itemize}
We should note that each of the stated algorithms is able to recognize if a given vertex $u$ of $G$ is avoidable in time $O(nm)$, $O(d(u)(n+m))$, and $O(n^{\omega})$, respectively, where $d(u)$ is the degree of $u$ in $G$.
Further, all of our proposed algorithms are characterized by their simplicitiy
and, besides the fast matrix multiplication, consist of basic ingredients that avoid using sophisticated data structures. 

In addition, we consider the natural generalizations of avoidable vertices, captured within the notions of the avoidable edges and avoidable paths.
A naive algorithm that recognizes an avoidable edge takes time $O(n^2 \cdot m)$ or $O(\overline{m} \cdot m)$.
Here we show that recognizing an avoidable edge of a graph $G$ can be done in $O(n \cdot m)$ time.
This is achieved by taking advantage of the notions of the $S$-excluded paths and their efficient detection by the modified breadth-first search algorithm.
Also notice that an avoidable edge is an avoidable path on two vertices.
We are able to reduce the problem of recognizing an avoidable path of arbitrary length to the recognition of an avoidable edge.
In particular, given an induced path we prove that we can replace the induced path by an edge and test whether the new added edge is avoidable or not in a reduced graph.
Therefore our recognition algorithm for testing whether a given induced path is avoidable takes $O(n \cdot m)$ time.
As a side remark of the later algorithm, we partially resolve an open question raised in \cite{Gurvich22}.
In particular, \cite{Gurvich22} asks whether their algorithm for identifying an avoidable path given an induced path, runs in polynomial time.
Our result implies that if the given path is avoidable then their algorithm runs in polynomial time.

\section{Preliminaries}
All graphs considered here are finite undirected graphs without loops and multiple edges.
We refer to the textbook by Bondy and Murty~\cite{Bondy} for any undefined graph terminology.
For a graph $G=(V_G, E_G)$, we use $V_G$ and $E_G$ to denote the set of vertices and edges, respectively.
We use $n$ to denote the number of vertices of a graph and use $m$ for the number of edges.
Given $x\in V_G$, we denote by $N_G(x)$ the neighborhood of $x$.
The \emph{degree} of $x$ is the number of edges incident to $x$, denoted by $d_G(x)$. That is, $d_G(x)=|N_G(x)|$.
The closed neighborhood of $x$, denoted by $N_G[x]$, is defined as $N_{G}(x)\cup\{x\}$.
For a set $X\subset V(G)$, $N_G(X)$ denotes the set of vertices in $V(G)\setminus X$ that have at least one neighbor in $X$. Analogously, $N_G[X]=N_G(X)\cup X$.
Given $X\subseteq V_G$, we denote by $G-X$ the graph obtained from $G$ by the removal of the vertices of $X$.
If $X=\{u\}$, we also write $G-u$. The \emph{subgraph induced by $X$} is denoted by $G[X]$, and has $X$ as its vertex set and $\{uv~|~u,v\in X\mbox{ and }uv\in E_G\}$ as its edge set.
For $R\subseteq E(G)$, $G\setminus R$ denotes the
graph~$(V(G), E(G)\setminus R)$, that is a subgraph of $G$. If $R=\{e\}$, we also write $G \setminus e$.

A {\it clique} of $G$ is a set of pairwise adjacent vertices
of $G$, and a {\it maximal clique} of $G$ is a clique of $G$
that is not properly contained in any clique of $G$. An
{\it independent set} of $G$ is a set of pairwise non-adjacent
vertices of $G$.
The induced path on $k \geq 2$ vertices is denoted by $P_k$ and the
induced cycle on $k\geq 3$ vertices is denoted by $C_k$.
For an induced path $P_k$, the vertices of degree one are called \emph{endpoints}.
A vertex $v$ is {\it universal} in $G$ if $N[v] = V(G)$ and $v$ is {\it isolated} if $N(v) = \emptyset$.
A vertex of degree one is called \emph{leaf}.
A graph is {\it connected} if there is a path between any
pair of vertices.
A {\it connected component} of $G$ is a maximal connected subgraph of $G$.
For any two vertices $x$ and $y$ of a connected graph there is an induced path having $x$ and $y$ as endpoints.
Given two vertices $u$ and $v$ of a connected graph $G$, a set $S \subset V_G$ is called \emph{$(u,v)$-separator} if $u$ and $v$ belong to different connected components of $G-S$.
We say that $S$ is a separator if there exist two vertices $u$ and $v$ such that $S$ is a $(u,v)$-separator.
For a set of finite graphs $\mathcal{H}$, we say that a graph $G$ is $\mathcal{H}$-free if $G$ does not contain an induced subgraph isomorphic to any of the graphs of $\mathcal{H}$.

The \emph{disjoint union} of two graphs $G$ and $H$, denoted by $G \cup H$, is the graph on vertex set $V(G) \cup V(H)$ and edge set $E(G) \cup E(H)$.
The \emph{complement} of $G$, denoted by $\overline{G}$, is the graph on vertex set $V(G)$ and edge set $\{uv \mid uv \notin E(G)\}$.
We say that a graph $G$ is \emph{co-connected} if $\overline{G}$ is connected. Moreover a \emph{co-component} of $G$ is a connected component of $\overline{G}$.

Given an edge $e=xy$, the {\it contraction} of $e$ removes both $x$ and $y$ and replaces them by a new vertex $w$, which is made adjacent to those vertices that were adjacent to at least one of the vertices $x$ and $y$, that is $N(w)=(N(x) \cup N(y))\setminus\{x,y\}$.
Let $S$ be a vertex set of $G$ such that $G[S]$ is connected.
If we repeatedly contract an edge of $G[S]$ until one vertex remains in $S$ then we say that we \emph{contract $S$ into a single vertex}.
In different terminology, {\it contracting a set of vertices} $S$ is the operation of substituting the vertices of $S$ by a new vertex $w$ with $N(w)=N(S)$.


A vertex $v$ is called \emph{simplicial} if the vertices of $N_G(v)$ induce a clique.
Listing all simplicial vertices of a graph can be done $O(nm)$ time.
The fastest algorithm for listing all simplicial vertices takes time $O(n^{\omega})$, where $O(n^{\omega})$ is the time needed to multiply two $n \times n$ binary matrices \cite{KloksKM00}
(currently, $\omega < 2.37286$ \cite{AlmanW21}).
Avoidable vertices and edges generalize the concept of simplicial vertices in a natural way.

\begin{definition}\label{def:avoidablevertex}
A vertex $v$ is called \emph{avoidable} if every $P_3$ with middle vertex $v$ is contained in an induced cycle.
Equivalently, $v$ is avoidable if $d_G(v) \leq 1$ or for every pair $x,y \in N_G(v)$
the vertices $x$ and $y$ belong to the same connected component of $G - (N_G[u] \setminus \{x,y\})$.
\end{definition}

Every simplicial vertex is avoidable, however the converse is not necessarily true.
It is known that every graph contains an avoidable vertex \cite{AboulkerCTV15,BerryB98,OHTSUKI1976622}.
Every vertex of a graph of degree $\le 1$ is simplicial and hence avoidable. Thus a non-avoidable vertex of a graph, has degree $\ge 2$.

\begin{obs}\label{obs:separtor}
Let $G$ be a graph and let $u$ be a vertex of $G$. Then $u$ is non-avoidable if and only if there is an $(x,y)$-separator $S$ that contains $u$ such that $S \subset N_{G}[u]$ for some vertices $x,y \in N_{G}(u)$.
\end{obs}
\begin{proof}
Assume that $u$ is non-avoidable. Then by Definition~\ref{def:avoidablevertex}, there are two vertices $x,y \in N_{G}(u)$ that belong to different connected components in $G - (N_G[u] \setminus \{x,y\})$.
This means that $S = N_G[u] \setminus \{x,y\}$ is an $(x,y)$-separator.
On the other hand, if there is such a separator $S$ for some vertices $x,y \in N_{G}(u)$ then $x$ and $y$ do not belong to the same connected component in the graph $G - S$ and, consequently, also in the graph $G - (N_G[u] \setminus \{x,y\})$, because $S \subset N_{G}[u]$. Thus $u$ is non-avoidable vertex.
\end{proof}


\subsection{A Lower Bound for Recognizing an Avoidable Vertex}
In the forthcoming sections, we give algorithms for recognizing an avoidable vertex in $O(nm)$ time and $O(n^{\omega})$ time.
Here we show that, under plausible complexity assumptions, a significant improvement on the stated running times is unlikely,
as we show that there is no truly subquadratic algorithm for deciding whether a given vertex is avoidable.
By \emph{truly subquadratic}, we mean an algorithm with running time $O(n^{2-\epsilon})$, for some $\epsilon >0$ where $n$ is the size of its input.

More precisely, the Strong Exponential-Time Hypothesis (SETH) states that
for any $\epsilon >0$, there exists a $k$ such that the $k$-SAT problem on $n$ variables cannot be solved in $O((2 - \epsilon)^n)$ time \cite{ImpagliazzoP01}.
The Orthogonal-Vector problem (OV) takes as input two families $A$ and $B$ of $n$ sets over a universe $C$, 
and asks whether there exist $a \in A$ and $b \in B$ such that $a\cap b = \emptyset$. An instance of OV is denoted by $OV(A,B,C)$.
It is known that under SETH, for any $\epsilon >0$, there exists a constant $c>0$ such that $OV(A,B,C)$ cannot be solved in $O(n^{2 - \epsilon})$, even if $|C| \leq c \cdot \log{n}$ \cite{Williams05}.
For deciding whether a given vertex is avoidable, we give a reduction from OV.

\begin{theorem}\label{theo:lowerbound}
The OV problem with $|A|=|B|=n$ can be reduced in $O(n \log{n})$ time to the problem of deciding whether a particular vertex of an $O(n)$-vertex graph is avoidable.
\end{theorem}
\begin{proof}
Let $OV(A,B,C)$ be an instance of OV.
We construct a graph $G$ as follows.
The vertex set of $G$ consists of $A \cup B \cup C$ and three additional vertices $u,c_A,c_B$. For the edges of $G$, we have:
\begin{itemize}
\item $u$ is adjacent to every vertex of $A \cup B$;
\item $c_A$ is adjacent to every vertex of $A$ and $c_B$ is adjacent to every vertex of $B$;
\item for every $a \in A$ and every $c \in C$, $ac \in E(G)$ if and only if $c \in a$;
\item for every $b \in B$ and every $c \in C$, $bc \in E(G)$ if and only if $c \in b$.
\end{itemize}
These are exactly the edges of $G$. In particular notice that $G[C\cup \{u,c_A,c_B\}]$ is an independent set.
Moreover, observe that $G$ has $2n+|C|+3$ vertices and the number of edges is $O(n \log{n})$. 
We claim that $OV(A,B,C)$ is a yes-instance if and only if $u$ is non-avoidable in $G$.

Assume that there are sets $a \in A$ and $b \in B$ such that $a \cap b = \emptyset$.
Let $x \in A$ and $y \in B$ be the vertices of $A$ and $B$ that correspond to $a$ and $b$, respectively.
By construction, $x$ and $y$ are non-adjacent in $G$.
Moreover, by construction, $x$ and $y$ have no common neighbor in $C$, as $a \cap b = \emptyset$.
Now notice that all neighbors of $x$ and $y$ that do not belong to $N_G[u]=A \cup B \cup \{u\}$ are in $C\cup\{c_A,c_B\}$ and $G[C\cup\{c_A,c_B\}]$ is an edgeless graph.
Thus $x$ and $y$ belong to different components in $G - (N_G[u] \setminus \{x,y\})$ and $u$ is non-avoidable in $G$.

For the converse, assume that $u$ is non-avoidable in $G$.
Since $N(u) = A \cup B$ there are vertices $x,y \in A \cup B$ such that $x$ and $y$ lie in different components in $G - (N_G[u] \setminus \{x,y\})$.
If both $x$ and $y$ belong to the same set $A$, then they have a common neighbor $c_A$ in $G - (N_G[u] \setminus \{x,y\})$ which is not possible.
Similarly, both $x$ and $y$ do not belong to $B$ due to vertex $c_B$.
Thus $x \in A$ and $y \in B$.
As there are  no edges in $G[C\cup\{c_A,c_B\}]$ we deduce that $x$ and $y$ have no common neighbor in $C$.
Hence there are sets in $A$ and $B$ that correspond to the vertices $x$ and $y$, respectively, that have no common element.
Therefore $OV(A,B,C)$ is a yes-instance.
\end{proof}

\section{Detecting Avoidable Vertices in Sparse or Dense Graphs}
Here we show how to compute efficiently all avoidable vertices on sparse or dense graphs.
In particular, for a graph $G$ on $n$ vertices and $m$ edges, we consider the cases in which $m < n$ (sparse graphs) or $\overline{m} < n$ (dense graphs), where $\overline{m} = |E(\overline{G})|$.
Our main motivation comes from the naive algorithm that lists all avoidable vertices in $O(n \cdot \overline{m} \cdot (n+m))$ time that takes advantage of the non-edges of $G$ \cite{BeisegelCGMS19}.
We will show that we can handle the non-edges in linear time, so that the running time of the naive algorithm can be written as $O(n^3 \cdot m)$.
For doing so, we consider the behavior of avoidable vertices on the complement of a graph by considering the connected components in both $G$ and $\overline{G}$.
Before reaching the details of our approach, we give a simple linear-time algorithm on the class of cographs,
since they can be totally decomposed by the corresponding operations.

\subsection{Appetizer: an optimal algorithm on cographs}
A graph $G$ is \emph{cograph} if every induced subgraph
of $G$ on at least two vertices is either disconnected or its complement is disconnected.
Cographs are exactly the class of $P_4$-free graphs \cite{CorneilLB81}.
Every cograph $G$ admits a unique tree representation known as \emph{cotree} which is a rooted tree $T$ with two types of internal nodes: 0-nodes and 1-nodes.
The vertices of $G$ are assigned to the leaves of $T$ in a one-to-one manner. Thus $T$ contains $O(n)$ nodes.
The properties of a cotree $T$ are summarized as follows:
\begin{itemize}
\item[(i)] Two vertices of $G$ are adjacent if and only if their least common ancestor in $T$ is a 1-node.
\item[(ii)] Every internal node of $T$ has at least two children.
\item[(iii)] No two internal nodes of the same type are adjacent in $T$.
\end{itemize}
The cotree of a cograph is unique and can be generated in linear time \cite{CorneilPS85}.

We give the following characterization of avoidable vertices in $G$ in terms of the cotree $T$.
For doing so, we denote by $p(u)$ the parent of a vertex $u$ in $T$.
A 1-node $w$ of $T$ is called \emph{full 1-node} if the children of $w$ are all leaves in $T$.
\begin{figure}[t]
\centering
\includegraphics[scale= 1.0]{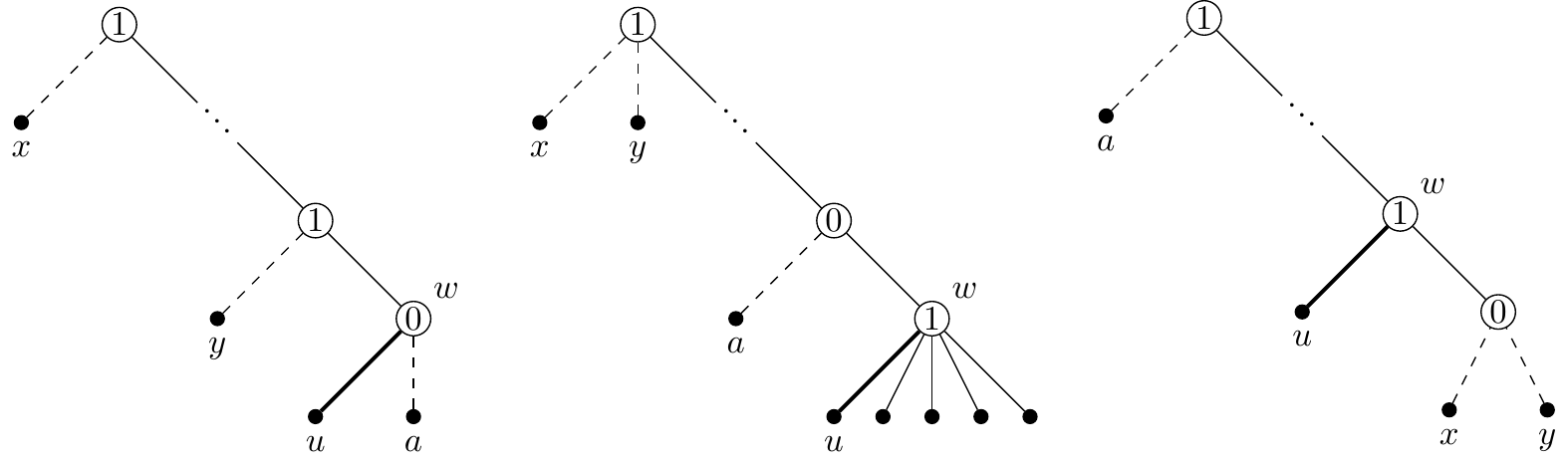}
\caption{Illustrating the cases considered in the proof of Lemma~\ref{char:cotree}.}\label{fig:cotree}
\end{figure}

\begin{lemma}\label{char:cotree}
Let $T$ be a cotree of a cograph $G$ and let $u$ be a vertex of $G$.
Then, $u$ is avoidable in $G$ if and only if either $p(u)$ is a 0-node or $p(u)$ is a full 1-node.
\end{lemma}
\begin{proof}
We first introduce some notation.
For a node $w$ of $T$, we let $T_w$ be the subtree of $T$ rooted at $w$ and we denote by $V(T_w)$ the set of leaves in $T_w$.
Recall that $V(T_w)$ corresponds to a subset of vertices of $G$.
By property~(i) observe that all the vertices of $V(T_w)$ are either adjacent or non-adjacent to a vertex $x$ of $V(G) \setminus V(T_w)$.
Let $r$ be the root of $T$ and let $w$ be the parent node of vertex $u$, that is $w = p(u)$.
We consider separately the following cases (see Figure~\ref{fig:cotree}).

\begin{itemize}
\item Assume that $w$ is a 0-node in $T$. We show that $u$ is avoidable in $G$.
Consider two vertices $x,y \in N_G(u)$.
By property~(i), $x,y \in V(T_r)\setminus V(T_w)$ and any vertex of $V(T_w)$ is non-adjacent to $u$.
Moreover, property~(ii) implies that there is a vertex $a \in V(T_w)\setminus \{u\}$ such that $au \notin E(G)$.
Thus, both $x$ and $y$ are adjacent to $u$ and $a$, since $x,y \notin V(T_w)$.
Hence, regardless of whether $x$ and $y$ being adjacent, there is a path between $x$ and $y$ that does not contain any vertex of $N_G(u)$.

\item Assume that $w$ is a full 1-node in $T$. We show that $u$ is avoidable in $G$.
Consider two vertices $x,y \in N_G(u)$. If $x \in V(T_w)$ then $xy \in E(G)$ because either $y \in V(T_w)$ as a leaf vertex, or $y \notin V(T_w)$ and $y$ is adjacent to every vertex of $V(T_w)$ as $uy \in E(G)$. Suppose that both $x,y \in V(T_r)\setminus V(T_w)$.
Let $P(r,w)$ be the unique path of $T$ between the root $r$ and the 1-node $w$.
Since $x,y \in N_G(u)$, there are 1-nodes $w_x$ and $w_y$ (not necessarily distinct) on $P(r,w)$ such that $x \in V(T_{w_x})$ and $y \in V(T_{w_y})$.
Now consider the parent $w'$ of $w$ in $T$. By property~(iii), $w'$ exists and is a 0-node of $T$.
Thus there is a vertex $a \in V(T_{w'})\setminus V(T_{w})$ that is non-adjacent to $u$.
Since the least common ancestor of $x$ and $a$ is $w_x$, by property~(i) we have $xa \in E(G)$. Similarly, we have $ya \in E(G)$.
Hence there is a path between $x$ and $y$ that contains a non-neighbor of $u$, which shows that $u$ is an avoidable vertex of $G$.

\item Assume that $w$ is a 1-node that is not full in $T$. We show that $u$ is non-avoidable in $G$.
Let $w'$ be a non-leaf child of $w$. By property~(iii), $w'$ is a 0-node.
Moreover, property~(ii) implies that there are vertices $x,y \in V(T_{w'})$ for which their least common ancestor is $w'$.
Thus $xy \notin E(G)$ and $ux,uy \in E(G)$, because $w$ is a 1-node.
If there is no path between $x$ and $y$ in $G - u$ then $u$ is non-avoidable.
Let $A$ be the internal vertices of an induced path between $x$ and $y$ in $G - u$.
Since $G$ is $P_4$-free, every vertex of $A$ is adjacent to both $x$ and $y$, so that $A=\{a\}$.
We show that $u$ is adjacent to $a$.
To see this, observe that $a$ does not belong to $V(T_{w'})$, since $w'$ is the 0-node that is the least common ancestor of $x$ and $y$.
Hence $a$ belongs to $V(T_{r}) \setminus V(T_{w'})$ and its least common ancestor $w_a$ with $x$ and $y$ is a 1-node.
This means that $w_a$ is an ancestor of $w'$ that is a 1-node in $T$.
As $w'$ is a child of $w$, we deduce that $w_a$ is the least common ancestor of $a$ and $u$.
Thus $ua \in E(G)$, which means that $u$ is non-avoidable, since there is no path between $x$ and $y$ that avoids any neighbor of $u$.
\end{itemize}
Therefore, we have a complete characterization of $u$ since all cases have been considered depending on the parent of $u$ in $T$.
\end{proof}

Thus, we deduce the following optimal algorithm for the vertices of a cograph $G$.
Note that, given a cograph $G$, its corresponding cotree $T$ can be constructed in $O(n+m)$ time \cite{CorneilPS85}.

\begin{theorem}\label{theo:cographs}
Given a cotree $T$ of a cograph $G$, there is an $O(n)$-time algorithm that lists all avoidable vertices of $G$.
\end{theorem}
\begin{proof}
We first mark the internal nodes of the cotree $T$ that have as children only leaves of $T$.
By a single bottom-up traversal from the leaves of $T$, this can be done in $O(n)$ time.
Thus applying Lemma~\ref{char:cotree} in a straightforward way on the cotree $T$ with the marked information, results in an $O(n)$-time algorithm.
\end{proof}

\subsection{Sparse or dense graphs}
Here we extend the previous notions on cographs and show how to handle the cases in which $m < n$ (sparse graphs) or $\overline{m} < n$ (dense graphs).

It is not difficult to handle sparse graphs. Observe that $m <n$ implies that $G$ is disconnected or $G$ is a tree. The connectedness assumption of the input graph $G$ follows from the fact that a vertex $u$ is avoidable in $G$ if and only if $u$ is avoidable in the connected component containing $u$, since there are no paths between vertices of different components. Moreover, trees have a trivial solution as the leaves are exactly the set of avoidable vertices. We include both properties in the following statement.

\begin{obs}\label{obs:connected}
Let $u$ be a vertex of $G$ and let $C(u)$ be the connected component of $G$ containing $u$. Then $u$ is avoidable if and only if $u$ is avoidable in $G[C(u)]$. Moreover, if $G$ is a tree then $u$ is avoidable if and only if $u$ is a leaf in $G$.
\end{obs}

Next we describe that we can follow almost the same approach on the complement of $G$.
For doing so, we first prove the following result which interestingly relates avoidability on $G$ and $\overline{G}$.
Note, however, that the converse is not necessarily true.

\begin{lemma}\label{lem:isolated}
Let $G$ be a graph and let $u$ be a non-avoidable vertex. Then, $u$ is avoidable in $\overline{G}$.
\end{lemma}
\begin{proof}
Since $u$ is a non-avoidable vertex in $G$, there is a separator $S$ that contains $u$ such that $S \subset N_{G}[u]$ by Observation~\ref{obs:separtor}.
Let $C_1, \ldots, C_k$ be the connected components of $G - S$, with $k \geq 2$.
Notice that at least two components of $C_1, \ldots, C_k$ contain a neighbor of $u$.
Without loss of generality, assume that $C_1 \cap N_{G}(u) \neq \emptyset$ and $C_2 \cap N_{G}(u) \neq \emptyset$.
Consider the complement $\overline{G}$ and let $x,y$ be two neighbors of $u$ in $\overline{G}$.
Observe that both $x$ and $y$ do not belong to $S$, since $S \subset N_{G}[u]$.
Thus $x \in C_i$  and $y \in C_j$, for $1 \leq i,j \leq k$.
We show that either $xy \in E(\overline{G})$ or there is a path in $\overline{G}$ between $x$ and $y$ that avoids vertices of $N_{\overline{G}}(u)$.
If $i \neq j$ then $xy \in E(\overline{G})$, because every vertex of $C_i$ is adjacent to every vertex of $C_j$ in $\overline{G}$.
Suppose that $x,y \in C_i$.
If $C_i \neq C_1$ then there is a vertex $w_1 \in C_1 \cap N_{G}(u)$ such that $w_1u \notin E(\overline{G})$ and $w_1x, w_1y \in E(\overline{G})$.
If $C_i = C_1$ then there is a vertex $w_2 \in C_2 \cap N_{G}(u)$ such that $w_2u \notin E(\overline{G})$ and $w_2x, w_2y \in E(\overline{G})$.
Thus in both cases there is a path of length two between $x$ and $y$ that avoids vertices $N_{\overline{G}}(u)$.
Therefore, $u$ is avoidable in $\overline{G}$.
\end{proof}

We next deal with the case in which $\overline{G}$ is disconnected. Notice that if $G=K_n$ then every vertex of $G$ is simplicial and thus avoidable.

\begin{lemma}\label{lem:disconne}
Let $G \neq K_n$, $u \in V(G)$, and let $\overline{C}(u)$ be the co-component containing $u$.
Then, $u$ is avoidable in $G$ if and only if $|\overline{C}(u)|>1$ and $u$ is avoidable in $G[\overline{C}(u)]$.
\end{lemma}
\begin{proof}
Assume first that $\overline{C}(u) = \{u\}$. Then $u$ is universal in $G$. Since $G \neq K_n$, there are vertices $x,y$ such that $xy \notin E(G)$. As any path between $x$ and $y$ contains a neighbor of $u$, we deduce that $u$ is non-avoidable. In the following we assume that $|\overline{C}(u)|>1$. This assumption implies that there is a vertex $a \in \overline{C}(u)$ such that $ua \notin E(G)$.
Also notice that every vertex of $\overline{C}(u)$ is adjacent to every vertex of $V(G) \setminus \overline{C}(u)$.

\begin{itemize}
\item Suppose that $u$ is avoidable in $G$. Assume for contradiction that $u$ is non-avoidable in $G[\overline{C}(u)]$. Then there are vertices $x,y$ in $\overline{C}(u)$ such that $x,y \in N_{G}(u)$, $xy \notin E(G)$, and every path (if it exists) between $x$ and $y$ in $G[\overline{C}(u)]$ contains a neighbor of $u$.
Since $G[\overline{C}(u)]$ is an induced subgraph of $G$ and $u$ is avoidable in $G$, there is path in $G$ between $x$ and $y$ that contains a vertex $z$ of $V(G) \setminus \overline{C}(u)$ such that $zu \notin E(G)$.
Then, however, we reach a contradiction to the fact that every vertex of $\overline{C}(u)$ is adjacent to every vertex of $V(G) \setminus \overline{C}(u)$, so that $zu \in E(G)$ for any such vertex $z$.
Thus $u$ is avoidable in $G[\overline{C}(u)]$.

\item Suppose that $u$ is avoidable in $G[\overline{C}(u)]$. We show that $u$ is avoidable in $G$. Consider two vertices $x,y \in N_{G}(u)$. If both vertices $x,y$ belong to $\overline{C}(u)$ then the avoidability of $u$ in $G[\overline{C}(u)]$ carries along $G$, since $G[\overline{C}(u)]$ is an induced subgraph of $G$. If $x \in \overline{C}(u)$ and $y \in V(G) \setminus \overline{C}(u)$ then $xy \in E(G)$. Now assume that both vertices $x,y$ belong to $V(G) \setminus \overline{C}(u)$. Then the path $\mypath{x,a,y}$ with $a \in \overline{C}(u)$ and $ua \notin E(G)$ is the desired path between $x$ and $y$. Thus $u$ is avoidable in $G$.
\end{itemize}
Therefore both directions show the claimed statement.
\end{proof}

In general, avoidability is not a hereditary property with respect to induced subgraphs, even when restricted to the removal of non-avoidable vertices.
However, as we show next, the removal of universal vertices does not affect the rest of the graph.

\begin{lemma}\label{lem:universal}
Let $G$ be a graph and let $w$ be a universal vertex of $G$. Then $w$ is avoidable if and only if $G$ is a complete graph.
Moreover, any vertex $u \in V(G)\setminus\{w\}$ is avoidable in $G$ if and only if $u$ is avoidable in $G-w$.
\end{lemma}
\begin{proof}
First statement follows by Lemma~\ref{lem:disconne} and from the fact that every vertex of a complete graph is simplicial.
For the second statement, assume that $u$ is avoidable in $G$.
We show that $u$ is avoidable in the graph $G' = G-w$.
Consider two vertices $x,y \in N_{G'}(u)$. If $xy \in E(G)$ then clearly $xy \in E(G')$.
Suppose that $xy \notin E(G)$. Then, as $u$ is avoidable in $G$, there is a path $P$ between $x$ and $y$ in $G$. Since $w$ is universal in $G$, $w$ does not belong to $P$.
Thus $P$ exists in $G'$ which shows that $u$ is avoidable in $G'$.
For the reverse direction, assume that $u$ is avoidable in $G'= G-w$.
Observe that any two vertices $x,y \in N_{G}(u) \setminus \{w\}$ fulfill the necessary conditions in $G$, since $G'$ is as induced subgraph of $G$.
Moreover, $w \in N_{G}(u)$ and for any vertex $x \in N_{G}(u) \setminus \{w\}$, we have $wx\in E(G)$.
Therefore $u$ remains avoidable in $G$.
\end{proof}

To conclude the cases for which $\overline{m} < n$, we next consider graphs whose complement is a tree.
By Observation~\ref{obs:connected} we restrict ourselves on connected graphs.

\begin{lemma}\label{lem:compltree}
Let $G$ be a connected graph such that $\overline{G}$ is a tree $T$.
A vertex $u$ of $G$ is avoidable if and only if $u$ is a non-leaf vertex in $T$.
\end{lemma}
\begin{proof}
We consider the vertices of $T$. Let $u$ be a non-leaf vertex of $T$. Then $u$ is a non-avoidable vertex in $\overline{G}$. Thus by Lemma~\ref{lem:isolated} $u$ is avoidable in $G$.

Now assume that $u$ is a leaf vertex of $T$, and thus avoidable in $\overline{G}$. We prove that $u$ is non-avoidable in $G$.
Since both graphs $G$ and $\overline{G}$ are connected, $u$ belongs to a $P_4$ in $T$ \cite{CorneilLB81}.
Let $\mypath{u,a,x,y}$ be a $P_4$ in $T$ that contains $u$.
Observe that $u$ is adjacent to every vertex of $V(G) \setminus \{a\}$ in $G$.
Consider the vertices $x$ and $y$ of the $P_4$ for which $x,y \in N_{G}(u)$.
As $xy \in E(\overline{G})$, we have $xy \notin E(G)$.
We show that there is no path between $x$ and $y$ that avoids any neighbor of $u$ in $G$.
If there is a path between $x$ and $y$ then it contains the vertex $a$ and it has the form $\mypath{x,a,y}$ in $G$.
Then, however, notice that $ya \in E(G)$ but $xa \notin E(G)$ by the induced $P_3=\mypath{a,x,y}$ in $\overline{G}$.
Thus $u$ is non-avoidable in $G$, because of $x$ and $y$.
Therefore, every avoidable vertex of $T$ is non-avoidable in $G$, since the set of leaves in $T$ are exactly the set of avoidable vertices in $\overline{G}$.
\end{proof}

Based on the previous results, we can reduce our problem to a graph $G$ that is both connected and co-connected and neither $G$ nor $\overline{G}$ are isomorphic to trees.
To achieve this in linear time we apply known techniques that avoid computing explicitly the complement of $G$, since we are mainly interested in recursively detecting the components and co-components of $G$.
Such a decomposition, known as the \emph{modular decomposition}, can be represented by a tree structure, denoted by $T(G)$, of $O(n)$ size and can be computed in linear time \cite{MCSP99,TedderCHP08}.
More precisely, the leaves of $T(G)$ correspond to the vertices of $G$ and every internal node $w$ of $T(G)$ is labeled with three distinct types according to whether the subgraph of $G$ induced by the leaves of the subtree rooted at $w$ is (i) not connected, or (ii) not co-connected, or (iii) connected and co-connected.
Moreover the connected components and the co-components of types (i) and (ii), respectively, correspond to the children of $w$ in $T(G)$.
Let $\mathcal{G}$ be a collection of maximal vertex-disjoint induced subgraphs of $G$ that are both connected and co-connected.
Then $T(G)$ determines all graphs of $\mathcal{G}$ in linear time.
Observe that if $\mathcal{G}$ is empty, then $G$ is a cograph.
%
%
In addition, we call $\mathcal{G}$, \emph{typical collection} of $G$ if for each graph $H \in \mathcal{G}$:
\begin{itemize}
\item $H$ is connected and co-connected,
\item $|V(H)| \leq |E(H)|$, $|V(H)| \leq |E(\overline{H})|$, and
\item every avoidable vertex in $H$ is an avoidable vertex in $G$.
\end{itemize}
The results of this section deduce the following algorithm.

\begin{theorem}\label{theo:mdavoid}
Let $G$ be a graph and let $A(G)$ be the set of avoidable vertices in $G$.
There is a linear-time algorithm, that
\begin{itemize}
\item computes a typical collection $\mathcal{G}$ of maximal vertex-disjoint induced subgraphs of $G$ and
\item for every vertex $v \in V(G) \setminus V(\mathcal{G})$, decides if $v \in A(G)$.
\end{itemize}
\end{theorem}
\begin{proof}
We first compute $T(G)$ in linear time \cite{MCSP99,TedderCHP08}.
Then we visit all nodes of $T(G)$ starting from the root and move towards the leaves of $T(G)$.
We stop each branch when we reach either a leaf for which we include it in $A(G)$, or when we reach a graph of $\mathcal{G}$.
Given a node $w$ of $T(G)$, let $G_w$ be the graph induced by the leaves of the subtree rooted at $w$.
At each node of $T(G)$ we perform the following steps.
\begin{enumerate}
\item If $G_w$ is disconnected then consider the connected components $C_1, \cdots, C_k$ of $G$ by Observation~\ref{obs:connected}.
That is, $A(G_w) = A(C_1) \cup \cdots \cup A(C_k)$.
\item If $\overline{G_w}$ is disconnected then consider the co-components $\overline{C}_1, \ldots, \overline{C}_k$ of $G$ such that $|\overline{C}_{i}| \geq 2$, for each $1\leq i \leq k$.
\begin{enumerate}
\item If $G_w=K_n$ (that is, $k=0$) then $A(G_w)=V(G_w)$.
\item Otherwise, $A(G_w) = A(\overline{C}_1) \cup \cdots \cup A(\overline{C}_k)$ by Lemma~\ref{lem:disconne}. Observe that all universal vertices in $G_w$ (that is, $|\overline{C}_{i}| = 1$) have been disregarded by Lemma~\ref{lem:universal}.
\end{enumerate}
\item Handling connected and co-connected graphs:
\begin{enumerate}
\item If $G_w=T$ then $A(G_w)=$ the set of leaves in $T$ by Observation~\ref{obs:connected}.
\item If $\overline{G_w}=T$ then $A(G_w)=$ the set of non-leaves in $T$ by Lemma~\ref{lem:compltree}.
\item Otherwise, include $G_w$ in the collection $\mathcal{G}$.
\end{enumerate}
\end{enumerate}
All steps can be carried out in $O(n+m)$ time by checking the type of the internal node $w$ in $T(G)$ and assigning the components and the co-components with the subtrees of $w$'s children.
Testing the corresponding cases whenever $G_w$ is connected and co-connected can be done by looking at the number of edges of $G_w$, that is in time $O(|V(G_w)|+|E(G_w)|)$.
Therefore the algorithm outputs in $O(n+m)$ time the described collection $\mathcal{G}$ and the set $A(G) \setminus A(\mathcal{G})$.
\end{proof}
%

\section{Computing Avoidable Vertices Directly from $G$}\label{sec:mintriang}
Here we give two different approaches for computing all avoidable vertices of a given graph $G$.
Both of them deal with the input graph itself without shrinking any unnecessary information, as opposed to the algorithms given in forthcoming sections.
Our first algorithm makes use of notions related to minimal triangulations of $G$ and runs in time $O(n^2 m)$.
The second algorithm runs in time $O(n^2 + m^2)$ and is based on a modified, traditional breadth-first search algorithm.

Let us first explain our algorithm through a minimal triangulation of $G$.
We first need some necessary definitions.
A graph is \emph{chordal} if it does not contain an induced cycle of length more than three.
In different terminology, $G$ is chordal if and only if $G$ is $(C_4, C_5, \ldots)$-free graph.

A graph $H=(V, E \cup F)$ is a \emph{minimal triangulation} of $G=(V,E)$ if $H$ is chordal and for every $F' \subset F$, the graph $(V,E \cup F')$ is not chordal.
The edges of $F$ in $H$ are called \emph{fill edges}.
Several $O(nm)$-time algorithms exist for computing a minimal triangulation \cite{Berry99,BerryBHP04,Heggernes06,RoseTL76}.
In connection with avoidable vertices, Beisegel et al. \cite{BeisegelCGMS19} showed the following characterization.

\begin{theorem}[\cite{BeisegelCGMS19}]\label{theo:charminBeisegel}
Let $u$ be a vertex of $G$. Then $u$ is avoidable in $G$ if and only if $u$ is a simplicial vertex in some minimal triangulation of $G$.
\end{theorem}

Although such a characterization is complete, it does not lead to an efficient algorithm for deciding whether a given vertex is avoidable,
since one has to produce \emph{all} possible minimal triangulations of $G$.
Here we strengthen such a characterization in the sense that it provides an efficient recognition based on a particular, \emph{nice}, minimal triangulation of $G$.

\begin{lemma}\label{lem:char_simpl}
Let $u$ be a vertex of a graph $G=(V,E)$ and let $H=(V,E\cup F)$ be a minimal triangulation of $G$ such that $u$ is not incident to any edge of $F$.
Then $u$ is avoidable in $G$ if and only if $u$ is simplicial in $H$.
\end{lemma}
\begin{proof}
If $u$ is simplicial in $H$ then by Theorem~\ref{theo:charminBeisegel} we deduce that $u$ is avoidable in $G$.
Suppose that $u$ is non-simplicial in $H$. Then there are two vertices $x,y \in N_G(u)$ that are non-adjacent in $H$.
Since $G$ is a subgraph of $H$, we have $xy \notin E(G)$.
We claim that there is no path in $G$ between $x$ and $y$ that avoids any vertex of $N_G[u] \setminus \{x,y\}$.
Assume for contradiction that there is such a path $P$. Then $V(P) \setminus \{x,y\}$ is non-empty and contains vertices only from $V \setminus N[u]$.
This means that $x, y$ belong to the same connected component of $H$ induced by $(V \setminus N[u]) \cup \{x,y\}$.
As $u$ is non-adjacent to any vertex of $V \setminus N[u]$ in $H$, the vertices of $(V \setminus N[u]) \cup \{x,y,u\}$ induce an induced cycle of length at least four in $H$.
Then we reach a contradiction to the chordality of $H$.
Therefore, there is no such path between $x$ and $y$, which implies that $u$ is non-avoidable in $G$.
\end{proof}

Next we show that such a minimal triangulation with respect to $u$, always exists and can be computed in $O(nm)$ time.
Our approach for computing a nice minimal triangulation of $G$ is \emph{vertex incremental},
in the following sense. We take the vertices of $G$ one by one in an arbitrary order $(v_1, \ldots, v_n)$,
and at step $i$ we compute a minimal triangulation $H_i$ of $G_i = G[\{v_1, \ldots, v_i\}]$
from a minimal triangulation $H_{i-1}$ of $G_{i-1}$ by adding only edges incident to $v_i$.
This is possible thanks to the following result.

\begin{lemma}[\cite{BerryHV06}]\label{lem:increm}
Let $G$ be an arbitrary graph and let $H$ be a minimal triangulation of $G$.
Consider a new graph $G' = G + v$, obtained by adding to $G$ a new vertex $v$.
There is a minimal triangulation $H'$ of $G'$ such that $H' - v = H$.
\end{lemma}

We denote by $H(v_1, \ldots, v_n)$ a vertex incremental minimal triangulation of $G$ which is obtained by considering the vertex ordering $(v_1, \ldots, v_n)$ of $G$.
Computing such a minimal triangulation of $G$, based on any vertex ordering, can be done in $O(nm)$ time \cite{BerryHV06}.

\begin{lemma}\label{lem:char_simpl_exist}
Let $u$ be a vertex of $G$ and let $X = N_{G}(u)$ and $A=V(G) \setminus N_{G}[u]$.
In any vertex incremental minimal triangulation $H(A,u,X)$ of $G$, no fill edge is incident to $u$.
\end{lemma}
\begin{proof}
Let $H(A,u,X) = (V, E \cup F)$ be a vertex incremental minimal triangulation of $G=(V,E)$.
Consider the vertex ordering $(A,u,X)$.
Observe that when adding $u$ to $H[A]$ no fill edge is required, as the considered graph $H[A] + u$ is already chordal.
Moreover $u$ is adjacent in $G$ to every vertex appearing after $u$ in the described ordering $(A,u,X)$.
Thus $u$ is non-adjacent to any vertex of $A$ in $H(A,u,X)$ which means that no edge of $F$ is incident to $u$.
\end{proof}

A direct consequence of Lemmas~\ref{lem:char_simpl} and \ref{lem:char_simpl_exist} is an $O(nm)$-time recognition algorithm for deciding whether a given vertex $u$ is avoidable.
For every vertex $u$, we first construct a vertex incremental minimal triangulation $H(A,u,X)$ of $G$ by applying the $O(nm)$-time algorithm given in \cite{BerryHV06}.
Then we simply check whether $u$ is simplicial in the chordal graph $H(A,u,X)$ by Lemma~\ref{lem:char_simpl}, which means that the overall running time is $O(nm)$.

\begin{algorithm}[t]
\SetKwInOut{KwIn}{Input}
\SetKwInOut{KwOut}{Output}
\KwIn{A graph $G$, a minimal triangulation $H$ of $G$, and a vertex $u$}
\KwOut{Returns true iff $u$ is avoidable in $G$}
Let $X = N_{G}(u)$ and $A=V(G) \setminus N_{G}[u]$\;  
Initialize a new graph $H' = H[A \cup \{u\}]$\;
Add the vertices of $X$ in $H'$ in an arbitrary order and maintain a minimal triangulation $H'$ of $G$ by applying the $O(nm)$-time algorithm given in \cite{BerryHV06}\;
\eIf{$u$ is simplicial in $H'$}
{\KwRet{true}\;}
{\KwRet{false}\;}
\caption{Testing if $u$ is avoidable with a vertex incremental minimal triangulation}
\label{algo:MTD}
\end{algorithm}

We note that one may compute any minimal triangulation $H$ of $G$, as a preprocessing step in time $O(nm)$, and
then use $H$ for constructing the vertex incremental minimal triangulation at each vertex $u$, so that $H[A]$ is already computed for $A=V(G)\setminus N_{G}[u]$.
Although such an approach results within the same theoretical time complexity, in practice it avoids recomputing common parts of the input data.
We give the details in Algorithm~\ref{algo:MTD} and, as already explained, its running time is $O(nm)$.
By applying Algorithm~\ref{algo:MTD} on each vertex, we obtain the following result.

\begin{theorem}\label{theo:algoMTD}
Listing all avoidable vertices of $G$ by using Algorithm~\ref{algo:MTD} takes $O(n^2 m)$ time.
\end{theorem}

An interesting remark of such an approach is that we can list all avoidable vertices of a chordal graph $G$ in an efficient way.
We note that such a result can be obtained directly from the definition of an avoidable vertex which shows that a non-simplicial vertex of a chordal graph is non-avoidable.

\begin{cor}\label{cor:chordal}
Let $G$ be a chordal graph. Listing all avoidable vertices of $G$ can be done in $O(n^{\omega})$ time, where $O(n^{\omega})$ is the time required to multiply two $n \times n$ binary matrices.
\end{cor}
\begin{proof}
By Lemma~\ref{lem:char_simpl} the set of simplicial vertices of $G$ is the set of avoidable vertices because any minimal triangulation $H$ of $G$ contains no fill edge, as $G$ is chordal.
Thus listing the avoidable vertices of a chordal graph $G$ reduces to listing the simplicial vertices of $G$.
Therefore detecting all avoidable vertices can be done in $O(n^{\omega})$ time by using the algorithm of \cite{KloksKM00}, which is the time needed to perform a fast matrix multiplication.
\end{proof}


\subsection{A fast algorithm for listing avoidable vertices}
Our second approach is based on the following notion of \emph{protecting} that we introduce here.
Given a set of vertices $S \subseteq V$, an $S$-excluded path is a path in which no internal vertex belongs to $S$.
Observe that an edge is an $S$-excluded path, for any choice of $S$.
By definition a single vertex is connected to itself by the trivial path.
Whenever there is an $S$-excluded path in $G$ between vertices $a$ and $b$, notice that $a$ can reach $b$ through vertices of $V(G) \setminus S$.

\begin{definition}[protecting]\label{def:protects}
Let $x$ and $y$ be two vertices of $G$.
We say that $x$ \emph{protects} $y$ if there is a $N_{G}[y]$-excluded path between $x$ and every vertex of $N_{G}(y)$.
In other words, $x$ protects $y$ if for any $z \in N_{G}(y)\setminus\{x\}$, either $xz \in E(G)$ or $x$ can reach $z$ through vertices of $V(G) \setminus N_{G}[y]$.
\end{definition}

Let us explain how to check if $x$ protects $y$ in linear time, that is in $O(n+m)$ time.
We consider the graph $G' = G - y$ and run a slight modification of a breadth-first search algorithm on $G'$ starting from $x$.
In particular, we try to reach the vertices of $N_{G}(y) \setminus \{x\}$ (target set) from $x$ in $G'$.
Every time we encounter a vertex $v$ of the target set, we include $v$ in a set $T$ of discovered target vertices and we do not continue the search from $v$ by avoiding to place $v$ within the search queue.
Consequently, no vertex of the target set is a non-leaf node of the constructed search tree.
Algorithm~\ref{algo:protectBFS} shows in detail the considered modification of a breadth-first search.

\begin{algorithm}[t]
\SetKwInOut{KwIn}{Input}
\SetKwInOut{KwOut}{Output}
\KwIn{A graph $G$, a vertex $x$, and a target set $S \subseteq V(G)$}
\KwOut{Returns true iff there is an $S$-excluded path between $x$ and every vertex of $S$}
Initialize a queue $Q=\{x\}$ and set $T=\emptyset$\;  
Mark $x$\;
\While{$Q$ is not empty}
{
  $s = Q.pop()$\;
  \For{$v \in N(s)$}{
    \If{$v$ is unmarked} {
      \eIf{$v \in S$} {
      $T = T \cup \{v\}$\;
      } {
      $Q.add(v)$\;
      }
      Mark $v$\;
    }
  }
}
\KwRet{$T = S$}
\caption{Detecting whether there is an $S$-excluded path between $x$ and every vertex of $S$}
\label{algo:protectBFS}
\end{algorithm}

\begin{lemma}\label{lem:algoSpath}
Algorithm~\ref{algo:protectBFS} is correct and runs in $O(n+m)$ time.
\end{lemma}
\begin{proof}
For the correctness, let $T$ be the search tree discovered by the algorithm when the search starts from $x$.
Observe that the basic concepts of the breadth-first search are maintained, so that the key properties with the shortest paths between the vertices of $G$ and the search tree $T$ are preserved.
If there is a leaf vertex $v$ in the constructed tree $T$ such that $v \in S$ then the unique path in $T$ is an $S$-excluded path in $G$ between $x$ and $v$,
since no vertex of $S$ is a non-leaf vertex of $T$.
On the other hand, assume that there is an $S$-excluded path in $G$ between $x$ and every vertex of $S$.
For every $v \in S$, among such $S$-excluded paths between $x$ and $v$, choose $P(v)$ to be the shortest.
Let $p(v)$ be the neighbor of $v$ in $P(v)$.
Clearly $x$ and every vertex $p(v)$ belong to the same connected component of $G$.
Consider the graph $G - S$.
Notice that every vertex $p(v)$ belongs to the same connected component with $x$ in $G - S$,
since for otherwise some vertices of $S$ separate $x$ and a vertex $v$ of $S$ which implies that there is no $S$-excluded path in $G$ between $x$ and $v$ in $G$.
Now let $T_x$ be a breadth-first search tree of $G-S$ that contains $x$.
Then the distance between $x$ and $p(v)$ in $T_x$ corresponds to the length of their shortest path in $G-S$.
Construct $T$ by attaching every vertex $v$ of $S$ to be a neighbor of $p(v)$ in $T_x$.
Therefore $T$ is a tree that contains the shortest $S$-excluded paths between $x$ and the vertices of $S$.

Regarding the running time, notice that no additional data structure is required compared to the classical implementation of the breadth-first search.
Hence the running time of Algorithm~\ref{algo:protectBFS} is bounded by the breadth-first search algorithm which is $O(n+m)$.
\end{proof}

Therefore we can check whether $x$ protects $y$ by running Algorithm~\ref{algo:protectBFS} on the graph $G - y$ with target set $S=N_{G}(y) \setminus \{x\}$.
The connection to the avoidability of a vertex, can be seen with the following result.

\begin{lemma}\label{lem:protect}
Let $u$ be a vertex of a graph $G=(V,E)$. Then $u$ is avoidable in $G$ if and only if $x$ protects $u$ for every vertex $x \in N_{G}(u)$.
\begin{proof}
Suppose first that $u$ is avoidable. Consider a vertex $x \in N_{G}(u)$.
Then for any vertex $y \in N_{G}(u)\setminus\{x\}$ there is a path between $x$ and $y$ that avoids vertices of $N_{G}(u)$.
This means that there is an $S$-excluded path between $x$ and $y$ with $S=N_{G}[u]$.
Thus $x$ protects $u$ in $G$.

For the other direction, assume that $u$ is non-avoidable.
Then there are vertices $x,y \in N_{G}(u)$ that belong to different connected components of $G - (N_G[u] \setminus \{x,y\})$.
Thus $x$ cannot reach $y$ through vertices of $V(G) \setminus N_{G}[u]$, implying that $x$ (and $y$) does not protect $u$.
Therefore there are at least two vertices in $N_{G}(u)$ that do not protect $u$.
\end{proof}
\end{lemma}

Now we are ready to show our fast algorithm for deciding whether a vertex is avoidable which is given in Algorithm~\ref{algo:PTD}.

\begin{algorithm}[t]
\SetKwInOut{KwIn}{Input}
\SetKwInOut{KwOut}{Output}
\KwIn{A graph $G$ and a vertex $u$}
\KwOut{Returns true iff $u$ is avoidable in $G$}
Let $X = N_{u}$ and $G' = G - u$\;
\For{$x \in X$}{
  Set $S = X \setminus\{x\}$\;
  \If{Algorithm~\ref{algo:protectBFS}($G',x,S$) is not true}
  {
    \KwRet{false}\;
  }
}
\KwRet{true}\;
\caption{Testing if $u$ is avoidable by detecting whether its neighbors protect $u$}
\label{algo:PTD}
\end{algorithm}

\begin{theorem}\label{theo:protect}
Listing all avoidable vertices of $G$ by using Algorithm~\ref{algo:PTD} takes $O(n^2+m^2)$ time.
\end{theorem}
\begin{proof}
Correctness follows from Lemmas~\ref{lem:algoSpath} and \ref{lem:protect}.
For the running time, observe that constructing $G'$ takes $O(n+m)$ time.
Moreover we need to make $d(u)$ calls to Algorithm~\ref{algo:protectBFS} for a particular vertex $u$ where $d(u)$ is the degree of $u$.
Thus, by Lemma~\ref{lem:algoSpath} the total running time is $O(\sum_{u} (1+d(u))(n+m)) = O(n^2+m^2)$.
\end{proof}


\section{Avoidable Vertices via Contractions}
Here we show how to compute all avoidable vertices of a graph $G$ through contractions.
Given a graph  $G=(V_G, E_G)$ and a vertex $u \in V_G$, we denote by $G_u$ the graph obtained from $G$ by contacting every connected component of $G - N_{G}[u]$.
We partition the vertices of $G_u - u$ into $(X,C)$, such that $X = N_{G}(u)$ and $C$ contains the contracted vertices of $G - N_{G}[u]$.
We denote by $G_{u}(X,C)$ the contracted graph where $(X,C)$ is the vertex partition with respect to $G_u$.
Observe that $G_{u}[X \cup \{u\}] = G[X \cup \{u\}]$ and $G_{u}[C \cup \{u\}]$ is an independent set.

\begin{obs}\label{obs:1contract}
Given a vertex $u$ of $G=(V,E)$, the construction of $G_u(X,C)$ can be done in $O(n+m)$ time.
\end{obs}
\begin{proof}
To compute the connected components $C_1, \ldots, C_k$ of $G - N_{G}[u]$ takes linear time.
For each vertex set $C_i$, $1 \leq i \leq k$, we compute $N_{G}(C_i)$ in time $d(C_i)$ where $d(C_i)$ is the sum of the degrees of the vertices in $C_i$.
As $C_1, \ldots, C_k$ is a partition of $V(G) \setminus N_{G}[u]$, the total running time for substituting each set $C_i$ is $O(k + \sum{d(C_i)}) = O(n+m)$.
\end{proof}

Next we show that $G_{u}(X,C)$ holds all necessary information of important paths of $G$ with respect to the avoidability of $u$.

\begin{lemma}\label{lem:2contract}
Let $u$ be a vertex of a graph $G=(V,E)$. Then $u$ is avoidable in $G$ if and only if $u$ is avoidable in $G_{u}(X,C)$.
\end{lemma}
\begin{proof}
Since $G[X \cup \{u\}] = G_{u}[X \cup \{u\}]$, we only need to consider the vertices of $X=N_{G}(u)$ that are non-adjacent.
Let $x,y \in N_{G}(u)=X$ such that $xy \notin E(G)$ and let $S = (X \cup \{u\})\setminus\{x,y\}$.
Observe that all vertices of $S$ belong to both graphs $G$ and $G_u(X,C)$.
We claim that
there is a path between $x$ and $y$ in $G - S$ if and only if
there is a path between $x$ and $y$ in $G_{u}(X,C) - S$.
Consider any path in $G - S$ of the form $\mypath{x,P,y}$.
The vertices of the given path belong to the same connected component of $G - S$.
Thus the vertices of $P$ belong to exactly one connected component $C_P$ of $G - (S\cup \{x,y\})$.
As $S\cup \{x,y\} = N_{G}[u]$, there is a vertex $C_i \in C$ that corresponds to $C_P$ in the contracted graph $G_{u}(X,C)$.
Hence, the path $\mypath{x,C_i,y}$ forms the desired path in $G_{u}(X,C) - S$.

If there is a path between $x$ and $y$ in $G_{u}(X,C) - S$ then such a path is of length two and has the form $\mypath{x,C_i,y}$ where $C_i \in C$.
Since $xC_i$ is an edge in $G_{u}(X,C)$, there is a vertex $a \in V(C_i)$ such that $xa \in E(G)$.
Similarly, there is a vertex $b \in V(C_i)$ such that $yb \in E(G)$.
As $a$ and $b$ belong to the same connected component $C_i$ of $G - N_{G}[u]$, there is a path $P_i$ in $G$ between $a$ and $b$ that contains only vertices from $V(C_i)$.
Thus there is a path $\mypath{x,P_i,y}$ in $G$ where $P_i \subseteq V(C_i)$.

Now observe that any path between two neighbors of $u$ in either $G - S$ or $G_{u}(X,C) - S$ does not contain any vertex of $N_{G}[u]$.
Therefore, by the above claim, we get the desired characterization of $u$ in both graphs.
\end{proof}

Lemma~\ref{lem:2contract} implies that we can apply all of our algorithms given in the previous section in order to recognize an avoidable vertex.
Although such an approach does not lead to faster theoretical time bounds, in practice the contracted graph has substantial smaller size than the original graph
and may lead to practical running times.
We next show that the contracted graph results in an additional algorithm with different running time.

Let $G_{u}(X,C)$ be the contracted graph of a vertex $u$. The \emph{filled-contracted graph}, denoted by $H_u(X,C)$, is the graph obtained from $G_{u}(X,C)$ by adding all necessary edges in order to make every neighborhood of $C_i \in C$ a clique. That is, for every $C_i \in C$, $N_{H_u}(C_i)$ is a clique.
The following proof resembles the characterization given through minimal triangulations in Lemma~\ref{lem:char_simpl}.
However observe that $H_u(X,C)$ is not necessarily a chordal graph, because $X \nsubseteq N_{G_u}(C)$.

%
%

\begin{lemma}\label{lem:conSimplicial}
A vertex $u$ is avoidable in $G$ if and only if $H_u[X]$ is a clique.
\end{lemma}
\begin{proof}
We apply Lemma~\ref{lem:2contract} and we need to show that $u$ is avoidable in $G_{u}(X,C)$ if and only if $H_u[X]$ is a clique.
Assume that $u$ is avoidable in $G_{u}(X,C)$. We show that $H_u[X]$ is a clique.
Consider two vertices $x,y \in X$.
If $xy$ is an edge in $G_{u}(X,C)$ then $xy$ remains an edge in $H_u(X,C)$, as $G_{u}(X,C)$ is a subgraph of $H_u(X,C)$.
If $x$ and $y$ are non-adjacent in $G_{u}(X,C)$, there is a vertex $C_i \in C$ such that $\{x,y\} \subseteq N_{G_{u}}(C_i)$, because $u$ is avoidable and $G_u[C]$ is an independent set.
Thus, by the definition of $H_u(X,C)$, $N_{H_u}(C_i)$ is a clique implying that $xy$ is an edge in $H_u[X]$.

Assume that $u$ is non-avoidable in $G_{u}(X,C)$. Then there are vertices $x,y \in X$ such that $xy \notin E(G_u)$ and they belong in different connected components of $G_u[C \cup \{x,y\}]$.
Thus $x$ and $y$ is a pair of non-adjacent vertices in $H_u[X]$, since there is no vertex $C_i \in C$ such that $x,y \in N_{G_u}(C_i)$.
Hence there is a pair of non-adjacent vertices in $H_u[X]$, so that $H_u[X]$ is not a clique.
\end{proof}

We take advantage of Lemma~\ref{lem:conSimplicial} in order to recognize whether $u$ is avoidable.
The naive construction of $H_u(X,C)$ requires $O(n^3)$ time, since $|X| \leq n$ and $|C|\leq n$.
Instead of constructing $H_u(X,C)$, we are able to check $H_u[X]$ in an efficient way through matrix multiplication.
To do so, we consider the graph $G'$ obtained from $G_u(X,C)$ by removing $u$ and deleting every edge with both endpoints in $X$.
Observe that the resulting graph $G'$ is a bipartite graph with bipartition $(X,C)$, as $G_{u}[C \cup \{u\}]$ is an independent set.
It turns out that it is enough to check whether two vertices of $X$ are in distance two in $G'$ which can be encapsulated by the square of its adjacency matrix.
Algorithm~\ref{algo:conMT} shows in details our proposed approach.

\begin{algorithm}[t]
\SetKwInOut{KwIn}{Input}
\SetKwInOut{KwOut}{Output}
\KwIn{A graph $G$ and a vertex $u$}
\KwOut{Returns true iff $u$ is avoidable in $G$}
Construct the contracted graph $G_{u}(X,C)$ of $u$\;
Let $G_1 = G_u(X,C) - u$\;
Construct the adjacency matrix $M_1$ of $G_1$\;
Let $G_2$ be the bipartite graph obtained from $G_1$ by removing every edge having both endpoints in $X$\;
Construct the adjacency matrix $M_2$ of $G_2$\;
Compute the square of $M_2$, i.e., $M_{2}^{2} = M_2 \cdot M_2$\;
Construct the matrix $M_3 = M_1 + M_{2}^{2}$\;
\For{$x,y \in X$}{
  \If{the entry $M_3[x,y]$ is zero}
  {
    \KwRet{false}\;
  }
}
\KwRet{true}\;
\caption{Testing if $u$ is avoidable by using matrix multiplication}
\label{algo:conMT}
\end{algorithm}

We are now in position to claim the following running time through matrix multiplication.

\begin{theorem}
Listing all avoidable vertices of $G$ by using Algorithm~\ref{algo:conMT} takes $O(n^{1+\omega})$ time, where $O(n^{\omega})$ is the time required to multiply two $n \times n$ binary matrices.
\end{theorem}
\begin{proof}
We apply Algorithm~\ref{algo:conMT} on each vertex of $G$.
Let us first discuss on the correctness of Algorithm~\ref{algo:conMT}.
By Lemma~\ref{lem:conSimplicial}, it is enough to show that $H_u[X]$ is a clique if and only if $M_3[X]$ has non-zero entries in its non-diagonal positions.
Let $G_1$ and $G_2$ be the two constructed graphs in Algorithm~\ref{algo:conMT}.
Observe that the square of $G_2$, denoted by $G_{2}^{2}$, is the graph obtained from the same vertex set of $G_2$ and two vertices $u,v$ are adjacent in $G_{2}^{2}$ if the distance of $u$ and $v$ is at most two in $G_2$.
Thus the matrix $M_{2}^{2}$ computed by Algorithm~\ref{algo:conMT} corresponds to the adjacency matrix of $G_{2}^{2}$.
Now it is enough to notice that two vertices $x,y$ of $X$ are adjacent in $H_u[X]$ if and only if $xy \in E(G_1) \cup E(G_{2}^{2})$.
In particular observe that if $x$ and $y$ have a common neighbor $w$ in $G_2$ then $w$ is a vertex of $C$ since there is no edge between vertices of $X$ in $G_2$ and $u \notin V(G_2)$.
Therefore $M_3[x,y]$ has a non-zero entry if and only if $x$ and $y$ are adjacent in $H_u[X]$.

Regarding the running time, notice that the construction of $G_u$ take linear time by Observation~\ref{obs:1contract}.
All steps besides the computation of $M_{2}^{2}$ can be done in $O(n^2)$ time.
The most time-consuming step is the matrix multiplication involved in computing $M_{2}^{2}$, which can be done in $O(n^{\omega})$ time.
Hence the total running time for recognizing all $n$ vertices takes $O(n^{1+\omega})$ time.
\end{proof}

\section{Recognizing Avoidable Edges and Paths}
Natural generalizations of avoidable vertices are avoidable edges and avoidable paths.
Here we show how to efficiently recognize an avoidable edge and an avoidable path.
Recall that the two vertices having degree one in an induced path $P_k$ on $k \geq 2$ vertices are called \emph{endpoints}.
Moreover, the edge obtained after removing the endpoints from an induced path $P_4$ on four vertices is called \emph{middle edge}.

\begin{definition}[simplicial and avoidable edge]\label{def:avoidableedge}
An edge $uv$ is called \emph{simplicial} if there is no $P_4$ having $uv$ as a middle edge.
An edge $uv$ is called \emph{avoidable} if either $uv$ is simplicial, or every $P_4$ with middle edge $uv$ is contained in an induced cycle.
\end{definition}

\noindent Given two vertices $x$ and $y$ of $G$, we define the following sets of the neighbors of $x$ and $y$:
\begin{itemize}
\item $B(x,y)$ contains the common neighbors of $x$ and $y$; i.e., $B(x,y)=N_{G}(x) \cap N_{G}(y)$.
\item $A_x$ contains the private neighbors of $x$; i.e., $A_x = N_{G}(x) \setminus (B(x,y) \cup \{y\})$.
\item $A_y$ contains the private neighbors of $y$; i.e., $A_y = N_{G}(y) \setminus (B(x,y) \cup \{x\})$.
\end{itemize}
Under this terminology, observe that $A_x \cap A_y = \emptyset$ and $N_{G}(\{x,y\})$ is partitioned into the three sets $B(x,y), A_x, A_y$.
Clearly all described sets can be computed in $O(d(x)+d(y))$ time.

\begin{obs}\label{obs:edgesimplicial}
An edge $xy$ of $G$ is simplicial if and only if $A_x = \emptyset$ or $A_y=\emptyset$ or every vertex of $A_x$ is adjacent to every vertex of $A_y$.
\end{obs}
\begin{proof}
Consider a $P_4=a,x,y,b$ that contains $xy$ as a middle edge.
Then $a \in A_x$ and $b \in A_y$ because $ay \notin E(G)$ and $xb \notin E(G)$.
Thus both sets $A_x$ and $A_y$ are non-empty.
Moreover, since $ab \notin E(G)$, we deduce that any non-edge with one endpoint in $A_x$ and the other in $A_y$ results in a $P_4$ having $xy$ as a middle edge.
\end{proof}

By Observation~\ref{obs:edgesimplicial}, the recognition of a simplicial edge can be achieved in $O(n+m)$ time:
consider the bipartite subgraph $H(A_x, A_y)$ of $G[A_x \cup A_y]$ which is obtained by removing every edge having both endpoints in either $A_x$ or $A_y$.
Then it is enough to check whether $H(A_x, A_y)$ is a complete bipartite graph.

We show that the more general concept of an avoidable edge can be recognized in $O(nm)$ time.
For doing so, we will take advantage of Algorithm~\ref{algo:protectBFS} and the notion of protecting given in Definition~\ref{def:protects}.
\begin{definition}[protected edge]\label{def:protectededge}
An edge $xy$ is \emph{protected} if there is an $(N_{G}[x] \cup N_{G}[y])$-excluded path between every vertex of $N_{G}(x)$ and every vertex of $N_{G}(y)$.
\end{definition}

\begin{figure}[t]
\centering
\includegraphics[scale= 1.4]{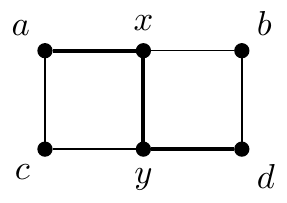}
\caption{In this example we have $N_{G}[x] \cup N_{G}[y] = V(G)$.
Observe that $x$ protects $y$, because $x$ has $\{c,y,d\}$-excluded paths to both $c$ and $d$, and similarly $y$ protects $x$.
However, the edge $xy$ is not protected because, for instance, there is no $V(G)$-excluded path (and, thus, an edge) between $a$ and $d$.
Also notice that there is a $P_4 = \mypath{a,x,y,d}$ that is not contained in an induced cycle.
}\label{fig:domino}
\end{figure}

We note that if an edge $xy$ is protected then $x$ protects $y$ and $y$ protects $x$ in accordance to Definition~\ref{def:protects}.
However, the reverse is not necessarily true, as shown in Figure~\ref{fig:domino}.

%

\begin{lemma}\label{lem:edgeprotected}
Let $xy$ be an edge of $G$. Then $xy$ is an avoidable edge in $G$ if and only if $xy$ is a protected edge in $G - B(x,y)$.
\end{lemma}
\begin{proof}
Let $H = G - B(x,y)$ and let us first show that $xy$ is an avoidable edge in $G$ if and only if $xy$ is an avoidable edge in $H$.
Suppose that $xy$ is an avoidable edge in $G$.
For any two vertices $a \in A_x$ and $b \in A_y$ such that $ab \notin E(G)$, there is an induced cycle $C$ that contains $a,x,y,b$.
Now observe that no vertex of $B(x,y)$ belongs to $C$, as $C$ is an induced cycle in $G$.
Thus $xy$ is an avoidable edge in $H$.
For the converse, notice that $H$ is an induced subgraph of $G$, so that all induced cycles of $H$ remain induced cycles in $G$.
Therefore our task is to show that $xy$ is an avoidable edge in $H$ if and only if $xy$ is protected in the same graph $H$.

Suppose that $xy$ is an avoidable edge in $H$. Observe that $N_{H}(x) = A_x \cup \{y\}$ and $N_{H}(y) = A_y \cup \{x\}$.
If at least one of $A_x, A_y$ is empty then $xy$ is protected (as well as simplicial),
because all required $(N_{H}[x] \cup N_{H}[y])$-excluded paths have length one between a vertex and its neighbors.
Consider any two vertices $a \in A_x$ and $b \in A_y$.
Clearly the edges $xa$ and $yb$ constitute $N_{H}[y]$-excluded path and $N_{H}[x]$-excluded path, respectively.
Assume first that $ab \notin E(H)$.
Then there is a $P_4=\mypath{a,x,y,b}$ that contains $xy$ as a middle edge.
Any induced cycle $C$ that contains the described $P_4$, contains vertices from $V(H) \setminus (A_x \cup A_y)$, so that the vertices of $C - P_4$ belong to $V(H) \setminus (N_{H}[x] \cup N_{H}[y])$.
Thus the subpath on $C$ taken from $C-P_4$ with endpoints $a$ and $b$ is a $(A_x \cup A_y \cup \{x,y\})$-excluded path of length at least two between $a$ and $b$.
If $ab \in E(H)$ then $\mypath{a,b}$ is an $(A_x \cup A_y \cup \{x,y\})$-excluded path of length one between $a$ and $b$.
In all cases we deduce that $xy$ is a protected edge.

Suppose that $xy$ is a protected edge in $H$.
Consider a $P_4=\mypath{a,x,y,b}$ that contains $xy$ as middle edge.
Then clearly $a \in A_x$, $b \in A_y$, and $ab \notin E(H)$.
We show that there is an induced cycle in $H$ that contains the $P_4$.
Between $a$ and $b$, there is an $(N_{H}[x] \cup N_{H}[y])$-excluded path $P_{ab}$ in $H$.
The length of $P_{ab}$ is at least two, since $ab \notin E(H)$.
By definition, all internal vertices of $P_{ab}$ belong to $V(H) \setminus (N_{H}[x] \cup N_{H}[y])$ and, thus, are non-adjacent to $x$ and $y$.
Let $S = V(P_{ab})$ and consider the induced subgraph $H[S]$ that is connected.
Then the shortest path $P'_{ab}$ between $a$ and $b$ in $H[S]$ is an induced path of $H$.
Therefore the concatenation of the $P_4=\mypath{a,x,y,b}$ with $P'_{ab}$ results in the desired induced cycle of $H$.
\end{proof}

Based on Lemma~\ref{lem:edgeprotected}, we deduce the following running time for recognizing an avoidable edge.
This is achieved by carefully applying Algorithm~\ref{algo:protectBFS}.
Notice that the stated running time is comparable to the $O(d(u)(n+m))$-time algorithm for recognizing an avoidable vertex $u$ implied by Theorem~\ref{theo:protect}.

\begin{theorem}\label{theo:avoidableedge}
Recognizing an avoidable edge of a graph $G$ can be done in $O(n \cdot m)$ time.
\end{theorem}
\begin{proof}
Let $xy$ be an edge of $G$. We first collect the vertices of $B(x,y)$ in $O(n)$ time.
By Lemma~\ref{lem:edgeprotected} we need to check whether $xy$ is protected in $H=G - B(x,y)$.
If $xy$ is simplicial edge then $xy$ is avoidable and, by Observation~\ref{obs:edgesimplicial}, this can be tested in $O(n+m)$ time.
Otherwise, both sets $A_x, A_y$ are non-empty.
Without loss of generality, assume that $|A_x| \leq |A_y|$.
In order to check if $xy$ is protected, we run $|A_x|$ times Algorithm~\ref{algo:protectBFS}:
\begin{itemize}
\item for every vertex $a \in A_x$, run Algorithm~\ref{algo:protectBFS} on the graph $(H-((A_x\setminus\{a\})\cup\{x,y\})$ started at vertex $a$ with a target set $A_y$.
\end{itemize}
In particular, we test whether there is an $A_y$-excluded path between $a$ and every vertex of $A_y$ without considering the vertices of $(A_x\setminus\{a\})\cup\{x,y\}$, that is on the graph $H-((A_x\setminus\{a\})\cup\{x,y\})$.
If all vertices of $A_x$ have an $A_y$-excluded path with all the vertices of $A_y$ on each corresponding graph, then such paths do not contain any internal vertex from $A_x \cup A_y \cup \{y\}$. Since $N_{H}[x]=A_x \cup \{x,y\}$ and $N_{H}[y]=A_y \cup \{x,y\}$, we deduce that $xy$ is a protected edge, and thus, $xy$ is avoidable in $G$.
Regarding the running time, observe that we make at most $n \geq |A_x|$ calls to Algorithm~\ref{algo:protectBFS} on induced subgraphs of $G$.
Therefore, by Lemma~\ref{lem:algoSpath}, the total running time is $O(nm)$.
\end{proof}

Let us now show how to extend the recognition of an avoidable edge towards their common generalization of avoidable induced paths.
The \emph{internal path} of a non-edgeless induced path $P$ is the path obtained from $P$ without its endpoints and its vertex set is denoted by $in(P)$.

\begin{definition}[simplicial and avoidable path]
An induced path $P_k$ on $k \geq 2$ vertices is called \emph{simplicial} if there is no induced path on $k+2$ vertices that contains $P_k$ as an internal path.
An induced path $P_k$ on $k \geq 2$ vertices is called \emph{avoidable} if either $P_k$ is simplicial, or every induced path on $k+2$ vertices that contains $P_k$ as an internal path is contained in an induced cycle.
\end{definition}

For $k=2$, avoidable paths correspond to avoidable edges. Let $P_k$ be an induced path on $k$ vertices of a graph $G$ with $k \geq 3$ having endpoints $x$ and $y$.
We denote by $I[P_k]$ the vertices of $N_{G}[in(P_k)]\setminus\{x,y\}$.
That is, $I[P_k]$ contains the vertices of the internal path of $P_k$ and their neighbors outside $P_k$.
Given two non-adjacent vertices $x$ and $y$ in $G$, we denote by $G+xy$ the graph obtained from $G$ by adding the edge $xy$.

\begin{lemma}\label{lem:avoidablepath}
Let $P_k$ be an induced path on $k$ vertices of a graph $G$ with $k \geq 3$ having endpoints $x$ and $y$.
Then $P_k$ is an avoidable path in $G$ if and only if $xy$ is an avoidable edge in $G + xy - I[P_k]$.
\end{lemma}
\begin{proof}
We claim first that
there is a $P_{k+2}$ that contains $P_k$ as an internal path in $G$ if and only if
there is a $P_4$ that contains $xy$ as a middle edge in the graph $H = G + xy - I[P_k]$.

Assume that there is a $P_{k+2}$ that contains $P_k$ as an internal path in $G$.
Let $x'$ and $y'$ be the endpoints of $P_{k+2}$.
As $P_{k+2}$ is an induced path, both $x',y'$ belong to $H$ and $x'y, xy', x'y' \notin E(H)$.
Thus $\mypath{x',x,y,y'}$ is a $P_4$ in $H$ that contains $xy$ as a middle edge.

Assume that there is a $P_4=\mypath{x',x,y,y'}$ in $H$ that contains $xy$ as a middle edge.
Consider the vertices of the path $P_{k-2}$ of $P_k -\{x,y\}$ in $G$ that correspond to the edge $xy$ of $H$.
Then no vertex of the $P_{k-2}$ is adjacent to any of $x'$ or $y'$ by the construction of $H$.
Thus, replacing the edge $xy$ in the $P_4=\mypath{x',x,y,y'}$ by the path $P_{k-2}$, results in an induced path $P_{k+2}$ on $k+2$ vertices in $G$.

Observe that the above claim implies that $P_k$ is a simplicial path in $G$ if and only if $xy$ is a simplicial edge in $H$.
Next we show that a non-simplicial path $P_k$ with endpoints $x$ and $y$ is avoidable in $G$ if and only if the non-simplicial edge $xy$ is avoidable in $H$.
Assume that there is a $P_{k+2} = \mypath{x',x,P_{k-2},y,y'}$ that contains $P_k=\mypath{x,P_{k-2},y}$ as an internal path in $G$.
Let $C_{G}$ be an induced cycle that contains the $P_{k+2}$ in $G$.
Since $C_{G}$ is induced cycle, every vertex of $C_{G} - P_{k-2}$ belongs to $H$.
Now observe that the vertices of $C_{G} - P_{k-2}$ induce a path in $G$ of length at least four.
Hence the vertices of $C_{G} - P_{k-2}$ induce a cycle in $H$, since $xy \in E(H)$, which shows that $xy$ is avoidable edge in $H$.

To show that $P_k$ is avoidable in $G$, we show that there is an induced cycle that contains the described $P_{k+2}$.
Let $C_{H}$ be an induced cycle of $H$ containing a $P_4=\mypath{x',x,y,y'}$.
Since $xy$ is a avoidable edge in $H$, such a cycle exists.
Construct the cycle $C'$ obtained from $C_{H}$ by removing the edge $xy$ and attaching the path $P_{k-2}$ of $P_k -\{x,y\}$.
Then $C'$ is an induced cycle in $G$ because:
\begin{itemize}
\item $C_{H} - \{x,y\}$ is an induced path in $G$, as $H- \{x,y\}$ is an induced subgraph of $G$,
\item $P_k$ is an induced path in $G$ by definition, and
\item no vertex of $P_{k-2}$ has a neighbor in $C_{H} - \{x,y\}$, as $N_{G}(P_{k-2})\setminus\{x,y\} \subset I[P_k]$.
\end{itemize}
Therefore there is an induced cycle in $G$ that contains the described $P_{k+2}$ of $P_k$.
\end{proof}

\begin{theorem}\label{theo:avoidablepath}
Given an induced path $P_k$ on $k>2$ vertices of $G$, testing whether $P_k$ is avoidable can be done in $O(n \cdot m)$ time.
\end{theorem}
\begin{proof}
Assume that the endpoints of $P_k$ are $x$ and $y$.
By Lemma~\ref{lem:avoidablepath}, it is enough to check if the edge $xy$ is avoidable in the graph $G + xy - I[P_k]$.
Constructing the graph $G + xy - I[P_k]$ takes $O(nk)$ time.
Applying the algorithm given in Theorem~\ref{theo:avoidableedge} results in an algorithm with the claimed running time, since $k \leq n$.
\end{proof}


\section{Concluding Remarks}
The running times of our algorithms for listing all avoidable vertices are comparable to the corresponding ones for listing all simplicial vertices.
Thus we believe it is difficult to achieve a reduction of the running time for avoidable vertices without affecting the time needed for simplicial vertices.
As pointed out, we can detect avoidable vertices in particular graph classes in more efficient way.
Towards this direction, it is interesting to consider planar graphs and reveal any possible improvement on the running time.
Moreover the notion of protecting and the relative $S$-excluded paths seem to tackle further problems concerning avoidable structures.
Our recognition algorithm for avoidable edges results in an algorithm for listing avoidable edges with running time $O(nm^2)$ which is comparable to the $O(m^2)$-algorithm for listing avoidable vertices.
Regarding avoidable paths on $k$ vertices, one needs to detect first with a naive algorithm a path $P_k$ in $O(n^k)$ time and then test whether $P_k$ being avoidable or not.
As observed in \cite{BonamyDHT20}, such a detection is nearly optimal, since we can hardly avoid the dependence of the exponent in $O(n^k)$.
Therefore by Theorem~\ref{theo:avoidableedge} we get an $O(n^{k+1} \cdot m)$-algorithm for listing all avoidable paths on $k$ vertices.

An interesting direction for further research along the avoidable paths is to reveal problems that can be solved efficiently by taking advantage the list of all avoidable paths in a graph.
For instance, one could compute a minimum length of a sequence of shifts transforming an induced path $P_k$ to an avoidable induced path. 
Gurvich et al. \cite{Gurvich22} proved that each induced path can be transformed to an avoidable one by a sequence of shifts,
where two induced paths on $k$ vertices are shifts of each other if their union is an induced path on $k + 1$ vertices.
To compute efficiently a minimum length of shifts, one could construct a graph $H$ that encodes all neighboring induced paths on $k$ vertices of $G$. 
In particular, the nodes of $H$ correspond to all induced paths on $k$ vertices in $G$ and two nodes in $H$ are adjacent if and only if their union is an induced path on $k + 1$ vertices in $G$. 
Note that $H$ contains $O(n^{k})$ nodes and can be constructed in $n^{O(k)}$ time. 
Having the list of avoidable paths on $k$ vertices, we can mark the nodes of $H$ that correspond to such avoidable paths. 
Now given an induced path $P_k$ on $k$ vertices in $G$ we may ask the shortest path in $H$ from the node that corresponds to $P_k$ towards a marked node that corresponds to an avoidable path.
Such a path always exists from the results of \cite{Gurvich22} and can be computed in time linear in the size of $H$. 
Therefore, for fixed $k$, our algorithm computes a minimum length of sequence of shifts in polynomial time answering an open question given in \cite{Gurvich22}.  

\bibliography{avoidable}

\begin{thebibliography}{10}

\bibitem{AboulkerCTV15}
Pierre Aboulker, Pierre Charbit, Nicolas Trotignon, and Kristina Vuskovic.
\newblock Vertex elimination orderings for hereditary graph classes.
\newblock {\em Discret. Math.}, 338(5):825--834, 2015.

\bibitem{AlmanW21}
Josh Alman and Virginia~Vassilevska Williams.
\newblock A refined laser method and faster matrix multiplication.
\newblock In {\em Proceedings of {SODA} 2021}, pages 522--539. {SIAM}, 2021.

\bibitem{BeisegelCGMS19}
Jesse Beisegel, Maria Chudnovsky, Vladimir Gurvich, Martin Milanic, and Mary
  Servatius.
\newblock Avoidable vertices and edges in graphs.
\newblock In {\em Proceedings of {WADS} 2019}, volume 11646, pages 126--139,
  2019.

\bibitem{Berry99}
Anne Berry.
\newblock A wide-range efficient algorithm for minimal triangulation.
\newblock In {\em Proceedings of {SODA} 1999}, pages 860--861. {ACM/SIAM},
  1999.

\bibitem{BerryBBS10}
Anne Berry, Jean R.~S. Blair, Jean~Paul Bordat, and Genevi{\`{e}}ve Simonet.
\newblock Graph extremities defined by search algorithms.
\newblock {\em Algorithms}, 3(2):100--124, 2010.

\bibitem{BerryBHP04}
Anne Berry, Jean R.~S. Blair, Pinar Heggernes, and Barry~W. Peyton.
\newblock Maximum cardinality search for computing minimal triangulations of
  graphs.
\newblock {\em Algorithmica}, 39(4):287--298, 2004.

\bibitem{BerryB98}
Anne Berry and Jean~Paul Bordat.
\newblock Separability generalizes dirac's theorem.
\newblock {\em Discret. Appl. Math.}, 84(1-3):43--53, 1998.

\bibitem{BerryHV06}
Anne Berry, Pinar Heggernes, and Yngve Villanger.
\newblock A vertex incremental approach for maintaining chordality.
\newblock {\em Discret. Math.}, 306(3):318--336, 2006.

\bibitem{BonamyDHT20}
Marthe Bonamy, Oscar Defrain, Meike Hatzel, and Jocelyn Thiebaut.
\newblock Avoidable paths in graphs.
\newblock {\em Electron. J. Comb.}, 27(4):P4.46, 2020.

\bibitem{Bondy}
J.~A. Bondy and U.~S.~R. Murty.
\newblock {\em Graph Theory}.
\newblock Springer, 2008.

\bibitem{CorneilLB81}
Derek~G. Corneil, H.~Lerchs, and L.~Stewart Burlingham.
\newblock Complement reducible graphs.
\newblock {\em Discret. Appl. Math.}, 3(3):163--174, 1981.

\bibitem{CorneilPS85}
Derek~G. Corneil, Yehoshua Perl, and Lorna~K. Stewart.
\newblock A linear recognition algorithm for cographs.
\newblock {\em {SIAM} J. Comput.}, 14(4):926--934, 1985.

\bibitem{Dirac61}
G.~A. Dirac.
\newblock On rigid circuit graphs.
\newblock {\em Abhandlungen aus dem Mathematischen Seminar der Universitat
  Hamburg}, 25(1):71--76, 1961.

\bibitem{Ducoffe22}
Guillaume Ducoffe.
\newblock The diameter of at-free graphs.
\newblock {\em J. Graph Theory}, 99:594--614, 2022.

\bibitem{Gurvich22}
Vladimir Gurvich, Matjaz Krnc, Martin Milanic, and Mikhail~N. Vyalyi.
\newblock Shifting paths to avoidable ones.
\newblock {\em Journal of Graph Theory}, 100:69--83, 2022.

\bibitem{Heggernes06}
Pinar Heggernes.
\newblock Minimal triangulations of graphs: {A} survey.
\newblock {\em Discret. Math.}, 306(3):297--317, 2006.

\bibitem{ImpagliazzoP01}
Russell Impagliazzo and Ramamohan Paturi.
\newblock On the complexity of k-sat.
\newblock {\em J. Comput. Syst. Sci.}, 62:367--375, 2001.

\bibitem{ItaiR78}
Alon Itai and Michael Rodeh.
\newblock Finding a minimum circuit in a graph.
\newblock {\em {SIAM} J. Comput.}, 7:413--423, 1978.

\bibitem{KloksKM00}
Ton Kloks, Dieter Kratsch, and Haiko M{\"{u}}ller.
\newblock Finding and counting small induced subgraphs efficiently.
\newblock {\em Inf. Process. Lett.}, 74(3-4):115--121, 2000.

\bibitem{KratschS06}
Dieter Kratsch and Jeremy~P. Spinrad.
\newblock Between {O(nm)} and o(n\({}^{\mbox{alpha}}\)).
\newblock {\em {SIAM} J. Comput.}, 36:310--325, 2006.

\bibitem{MCSP99}
R.~M. McConnell and J.~P. Spinrad.
\newblock Modular decomposition and transitive orientation.
\newblock {\em Discrete Mathematics}, 201:189--241, 1999.

\bibitem{OHTSUKI1976622}
Tatsuo Ohtsuki, Lap~Kit Cheung, and Toshio Fujisawa.
\newblock Minimal triangulation of a graph and optimal pivoting order in a
  sparse matrix.
\newblock {\em Journal of Mathematical Analysis and Applications},
  54(3):622--633, 1976.

\bibitem{RodittyW13}
Liam Roditty and Virginia~Vassilevska Williams.
\newblock Fast approximation algorithms for the diameter and radius of sparse
  graphs.
\newblock In {\em Proceedings of {STOC} 2013}, pages 515--524, 2013.

\bibitem{RoseTL76}
Donald~J. Rose, Robert~Endre Tarjan, and George~S. Lueker.
\newblock Algorithmic aspects of vertex elimination on graphs.
\newblock {\em {SIAM} J. Comput.}, 5(2):266--283, 1976.

\bibitem{TedderCHP08}
Marc Tedder, Derek~G. Corneil, Michel Habib, and Christophe Paul.
\newblock Simpler linear-time modular decomposition via recursive factorizing
  permutations.
\newblock In {\em Proceedings of {ICALP} 2008}, volume 5125 of {\em Lecture
  Notes in Computer Science}, pages 634--645, 2008.

\bibitem{Williams05}
Ryan Williams.
\newblock A new algorithm for optimal 2-constraint satisfaction and its
  implications.
\newblock {\em Theor. Comput. Sci.}, 348:357--365, 2005.

\end{thebibliography}

\end{document}